\documentclass[AMA,STIX1COL,doublespace]{WileyNJD-v2}

\usepackage{graphicx}
\usepackage{caption}
\usepackage{subcaption}
\usepackage{footnote}
\newcommand*\xbar[1]{%
	\hbox{%
		\vbox{%
			\hrule height 0.5pt 
			\kern0.5ex
			\hbox{%
				\kern-0.1em
				\ensuremath{#1}%
				\kern-0.1em
			}%
		}%
	}%
}

\articletype{RESEARCH ARTICLE}%

\received{29 January 2021}
\revised{14 April 2021}
\accepted{27 May 2021}

\begin{document}

\title{A practical Response Adaptive Block Randomization (RABR) design with analytic type I error protection}

\author[1]{Tianyu Zhan*}

\author[2]{Lu Cui}

\author[1]{Ziqian Geng}

\author[3]{Lanju Zhang}

\author[1]{Yihua Gu}

\author[4]{Ivan S.F. Chan}

\authormark{Zhan \textsc{et al}}

\address[1]{\orgdiv{Data and Statistical Sciences}, \orgname{AbbVie Inc.}, North Chicago, \orgaddress{\state{Illinois}, \country{USA}}}

\address[2]{\orgdiv{Statistical Sciences and Innovation}, \orgname{UCB Biosciences Inc.}, Raleigh, \orgaddress{\state{North Carolina}, \country{USA}}}

\address[3]{\orgdiv{Biometrics, Vertex Pharmaceuticals, Boston, Massachusetts, USA}}

\address[4]{\orgdiv{Global Biometrics and Data Sciences, Bristol Myers Squibb, Berkeley Heights, New Jersey, USA}}

\corres{*Tianyu Zhan, 1 Waukegan Rd, North Chicago, IL 60064, USA.  \email{tianyu.zhan.stats@gmail.com}}


\abstract[Summary]{Response adaptive randomization (RAR) is appealing from methodological, ethical, and pragmatic perspectives in the sense that subjects are more likely to be randomized to better performing treatment groups based on accumulating data. However, applications of RAR in confirmatory drug clinical trials with multiple active arms are limited largely due to its complexity, and lack of control of randomization ratios to different treatment groups. To address the aforementioned issues, we propose a Response Adaptive Block Randomization (RABR) design allowing arbitrarily pre-specified randomization ratios for the control and high-performing groups to meet clinical trial objectives. We show the validity of the conventional unweighted test in RABR with a controlled type I error rate based on the weighted combination test for sample size adaptive design invoking no large sample approximation. The advantages of the proposed RABR in terms of robustly reaching target final sample size to meet regulatory requirements and increasing statistical power as compared with the popular Doubly Adaptive Biased Coin Design (DBCD) are demonstrated by statistical simulations and a practical clinical trial design example.}

\keywords{confirmatory adaptive design; multi-arm studies; sample size; type I error rate control; unweighted statistics}


\maketitle


\section{Introduction}

Randomized clinical trials (RCTs) remain the gold standard for understanding the effect of a treatment or other intervention relative to placebo or standard of care. \citep{dia2015, wu2017, bar2018} To make the trial more efficient and ethical, confirmatory adaptive designs become more widely acceptable and popular to allow flexible interim modifications in the ongoing trials using unblinded data without compromising the type I error rate. \citep{fda2019, usc2016, bre2009} An appealing branch is the response adaptive randomization (RAR) design, where the chance of a newly enrolled subject being assigned to a treatment arm varies over the course of the trial based on accumulating outcome data for subjects previously enrolled. \citep{fda2019} Statistical, ethical, and pragmatic rationales support the advantage of using RAR as more subjects are being assigned to the more promising treatment arms. As a popular method of RAR, Doubly Adaptive Biased Coin Design (DBCD) accounts for both the current sample proportion and a desired allocation function. \citep{eis1994, hu2006, wil2020} Covariate-adjusted RAR designs are also receiving more attention recently. \citep{hu2015, zhu2018} 



However, in industry-sponsored confirmatory clinical studies to evaluate drug efficacy and safety, randomization ratios and targeted sample sizes of treatment groups are typically pre-specified in study protocols prior to conducting the clinical trial to meet regulatory requirements. \cite{fda1996} As demonstrated in Section \ref{s:sim} on simulations studies, although DBCD allows one to adjust tuning parameters in allocation functions based on assumptions at the study design stage, DBCD may potentially miss an agreed randomization target if observed data deviate. This is more challenging when there is limited prior knowledge on the efficacy of treatment regimes, as in the case study we considered in Section \ref{s:real}.



To address this challenge, we propose a Response Adaptive Block Randomization (RABR) design motived by Wang and Cui, \cite{wang2007} and Cui et al., \cite{cui2020} where subjects are adaptively randomized to different treatment groups based on the order of their standardized effect sizes and a pre-specified block randomization vector. As shown in simulation studies in Section \ref{s:sim} and the case study in Section \ref{s:real}, our RABR is more robust than DBCD in achieving a desired final sample size per group under varying underlying treatment effects. By effectively increasing the sample size of the selected treatment group to target levels, and decreasing the number of patients randomized to worse performing groups, our proposed RABR also has a higher power of detecting a significant treatment effect than DBCD. Moreover, the built-in block randomization feature of RABR facilitates its implementation in practice, in the sense that Interactive Response Technology (IRT) schedules can be specified in advance. While the application of RAR in blocks has been studied in some previous articles in the context of two-group comparisons with a binary endpoint, \citep{mag2011, kar2003} our approach is proposed to accommodate multi-arm trials with broader data types and to achieve targeted final sample size for the selected treatment group. Another contribution of our work is to analytically prove that the one-sided type I error rate for pairwise comparison is controlled at a nominal significance level by using the usual unweighted test statistic, which does not rely on large sample approximation despite the data-driven randomization process. The unweighted statistics are easier to compute than weighted statistics, \citep{cui2020, cui1997, cui1999} and they are a good alternative to the weighted combination test \citep{cui2020} particularly when multiple treatment groups are presented and complicated multiplicity adjustment is needed.



The remainder of this article is organized as follows. In Section \ref{s:dbcd}, we review the DBCD with a continuous endpoint, which will be used to compare with the proposed design and analysis. In Section \ref{s:design}, we introduce our proposed RABR. The analytic type I error rate protection in pairwise comparison is proved in Section \ref{s:inf}. Simulations are performed to evaluate type I error and power under various scenarios in Section \ref{s:sim}. As a practical example, we further re-design a confirmatory trial for treating tuberous sclerosis complex with RABR at Section \ref{s:real}. Concluding remarks are provided in Section \ref{s:dis}. 

\section{Review of the Doubly Adaptive Biased Coin Designs (DBCD)}

\label{s:dbcd}

In this section, we review the Doubly Adaptive Biased Coin Designs (DBCD). \cite{hu2006} Let us consider a clinical trial with a continuous endpoint, where the response $X^{\left(g\right)}_i$ for subject $i$ in treatment group $g$ follows independent normal distribution with mean $\mu_g$ and variance $\sigma_g^2$,
\begin{equation}
	\label{notation_x_normal}
	X^{\left(g\right)}_i \sim \mathcal{N}\left(\mu_g, \sigma_g^2\right),
\end{equation}
for $g=0$ labeling placebo and $g=1, \ldots, m$ for $m$ treatment groups. We assume that a higher response is favorable, and the response of each subject is observed right after its enrollment in the study. 

After a burn-in period where $M$ subjects have been equally randomized to $m+1$ groups, the study adjusts the randomization probability to each treatment group $g$ ($g = 0, 1, \ldots, m$) based on the following allocation function:
\begin{align}
	& a_g(\eta, 0, \tau_g) = 1, \quad a_g(\eta, 1, \tau_g) = 0, \nonumber \\
	& a_g(\eta, \theta_g, \tau_g) = \left[\tau_g\left(\frac{\tau_g}{\theta_g}\right)^\eta\right]\Bigg/\left[\sum_{j=0}^m \tau_j\left(\frac{\tau_j}{\theta_j}\right)^\eta\right] \label{a_g_original},
\end{align}
where $\eta$ is a non-negative tuning parameter controlling randomness of randomization probabilities. The asymptotic allocation variance is a monotone decreasing function of $\eta$, and usually $\eta=2$ is chosen in practice. \citep{hu2006} In $a_g(\eta, \theta_g, \tau_g)$, $\theta_g$ is the proportion of subjects that have been randomized to group $g$ so far, and $\tau_g$ is the targeting allocation probability for group $g$. In this article, we consider $\tau_g$ in the following form:
\begin{equation}
	\label{y_g_original}
	\tau_g = \left(\frac{1}{\sigma_g}\left[\Phi\left(\frac{\mu_g-\lambda}{\sigma_g}\right)\right]^{1/2}\right)\Bigg/\left(\sum_{j=0}^m \frac{1}{\sigma_j} \left[\Phi\left(\frac{\mu_j-\lambda}{\sigma_j}\right)\right]^{1/2}\right),
\end{equation} 
which is a direct generalization of the allocation function with a continuous endpoint of treatment failure in two-group comparison \citep{bis2004} to a study with multiple treatment groups. The tuning parameter $\lambda$ needs to be pre-specified. 

We further use $\widehat{\tau_g}$ to denote an estimate of $\tau_g$ by replacing $\mu_g$ and $\sigma_g$ by their consistent estimators. At each interim checkpoint, the adaptive randomization probability for treatment group $g$ is $a_g(\eta, \theta_g, \widehat{\tau_g})$. In the context of two-group comparison, it can be shown that $\widehat{\tau_g}$ and $\theta_g$ converge to $\tau_g$ almost surely, and with an asymptotic bivariate normal distribution under certain regulatory conditions. \citep{hu2006} When the sample size is large enough, the study reaches stabilized randomization probabilities to each group $g$ at $\tau_g$. Therefore, asymptotically, a level $\alpha$ test for a non-adaptive design is also valid to control the pairwise comparison error rate at $\alpha$. One can further use certain multiple test procedures to control the family-wise error rate (FWER).

Challenges exist for this type of design to robustly achieve a desired sample size for the selected arm and placebo with a given $\lambda$ under different underlying responses. This issue is further studied in Section \ref{s:sim}. Moreover, block randomization is generally not available given the functional form of $a_g$ in (\ref{a_g_original}). The trial may end up with an undesired final proportion per group, especially when the total sample size is relatively small. In the next section, we introduce our proposed Response Adaptive Block Randomization (RABR) Design as an alternative option along with an analytic type I error protection. 

\section{Response Adaptive Block Randomization (RABR) Design}

\label{s:design}

%

As a starting point, let us consider a simple two-stage study design with one interim adaptation on the randomization. The generalization to multiple stages is illustrated in the following Theorem \ref{theorem_error_multiple}. In stage 1, there are $N_1^{\left(g\right)}$ subjects that have been recruited to group $g$ with sample mean $\xbar{X}_1^{\left(g\right)} = \sum_{i=1}^{N_1^{\left(g\right)}} X_i^{\left(g\right)} / N_1^{\left(g\right)}$, for $g=0, 1, \ldots, m$. Further denote $\boldsymbol{X}_1^{\left(g\right)} = \left\{X_1^{\left(g\right)}, \ldots, X_{N_1^{\left(g\right)}}^{\left(g\right)}\right\}$ as the vector of responses in group $g$ at stage 1. Similar to the DBCD as described in Section \ref{s:dbcd}, we utilize stage 1 as a burn-in period, where equal randomization is adopted on the placebo and $m$ treatment group: $N_1^{\left(0\right)} = N_1^{\left(1\right)} = \cdots = N_1^{\left(m\right)}$. The burn-in size $M$ is equal to $\sum_{g=0}^m N_1^{\left(g\right)}$. Correspondingly, $N_2^{\left(g\right)}$ and $\xbar{X}_2^{\left(g\right)}$ are notations for sample size and sample mean in group $g$ at stage 2. We use $\xbar{X}^{\left(g\right)}$ and $N^{\left(g\right)}$ to denote the sample mean and sample size of group $g$ in the whole study, for $g = 0, 1, \cdots, m$. In this article, we consider studies with either $m=2$ or $3$ active treatment groups with generalization to $m>3$, as briefly discussed in Section \ref{s:dis}. 

Having observed the first stage data $\boldsymbol{D}_1 = \left\{\boldsymbol{X}_1^{\left(0\right)}, \boldsymbol{X}_1^{\left(1\right)}, \ldots, \boldsymbol{X}_1^{\left(m\right)}\right\}$, we adjust the randomization ratios of $m$ treatment groups in a data-driven manner by assigning more subjects to the more promising dose(s). Let vector $\boldsymbol{r} = (r_0, r_1, r_2, \ldots, r_m)$ denote a pre-specified block randomization vector for our design, where $r_g$ are integers satisfying $r_1 \geq r_2 \geq \cdots \geq r_m$ for $g = 1, \ldots, m$. The block size $B$ is equal to $\sum_{g=0}^m r_g$. The first element $r_0$ in $\boldsymbol{r}$ specifies the number of subjects being randomized to the placebo group in one block. Therefore, the randomization probability to the placebo is fixed at $r_0$ out of $B$. The rest of the $m$ elements in $\boldsymbol{r}:$ $r_1, \ldots, r_m$ are assigned to $m$ treatment groups based on a decreasing order of the standardized response denoted as $R_g$,
\begin{equation}
	\label{notation_measure}
	R_g = \left\{N_1^{\left(g\right)}\right\}^{1/2} \: \xbar{X}_1^{\left(g\right)} /\widehat{\sigma}_g, 
\end{equation}
where $\widehat{\sigma}_g$ is the sample standard deviation of data $\boldsymbol{X}_1^{\left(g\right)}$, for $g = 1, \ldots, m$. That is to say, the group with the largest standardized observed mean in the first stage gets the highest probability of being randomized in the second stage ($r_1$ out of $B$). The functional form in (\ref{notation_measure}) is appealing as it is on the same scale of the final test statistics, and it also eases our analytic proof of pairwise type I error rate control at Section \ref{s:con_error}. Even if $\boldsymbol{r}$ is pre-specified, its assignment to each treatment group $g$, $g=1, \ldots, m$ is dependent on the previous stage data $\boldsymbol{D}_1$. 

For demonstration, let's consider a study with $m=2$ active doses with $g = 0, 1, 2$ denoting the placebo, the low dose and the high dose group, respectively. Without loss of generality, we focus on the low dose $g=1$ for illustration. We use $N_2^{\left(1\right)}(\boldsymbol{D}_1)$ to denote the sample size in treatment group $g=1$ at the second stage. It takes the value $\beta_1 N_2^{\left(0\right)}$ if $\left\{N_1^{\left(1\right)}\right\}^{1/2} \:\: \xbar{X}_1^{\left(1\right)} / \widehat{\sigma}_1 \geq \left\{N_1^{\left(2\right)}\right\}^{1/2} \:\: \xbar{X}_1^{\left(2\right)} / \widehat{\sigma}_2$, and $\beta_2 N_2^{\left(0\right)}$ otherwise, where $\beta_1 = r_1/r_0$ and $\beta_2 = r_2 / r_0$. Essentially, subjects are more likely to be assigned to the low dose in the second stage if its response is more promising than that in the high dose based on $\boldsymbol{D}_1$. If $\beta_1 = \beta_2$, then it degenerates to a non-adaptive design with equal randomization. We simplify the notation of $N_2^{\left(1\right)}(\boldsymbol{D}_1)$ as $N_2^{\left(1\right)}$. 

To illustrate randomization procedures in RABR following some related discussion in Cui et al.,\cite{cui2020} we consider a study to evaluate low dose $L$ and high dose $H$ versus placebo $P$ with $\boldsymbol{r} = (2, 3, 1)$ of block size $B = 6$ in Figure \ref{F:RAR_flowchart}, where 2 out of 6 is for placebo and 3 out of 6 is for the better performing group. Before the current trial conduct, we first build three sets of randomization schedules: Schedule $Q_1$ with equal randomization probabilities, Schedule $Q_2$ with a higher randomization probability for $L$, and Schedule $Q_3$ with high probability for $H$. During the burn-in period, Schedule $Q_1$ is activated to achieve equal randomization. For each subject enrolled after burn-in period, the standardized responses $R_L$ and $R_H$ in (\ref{notation_measure}) are calculated based on accumulated data for $L$ and $H$, respectively. If $R_L \geq R_H$ which means that $L$ is performing better, then Schedule $Q_2$ is activated with a higher probability of $3$ out $6$ to assign the next subject to this dose, and vice versa. In our proposed RABR, there are several design parameters involved, for example the burn-in size $M$, block size $B$, allocation vector $\boldsymbol{r}$ and total sample size $N$. Their choices are relatively flexible, but have to be pre-specified in the protocol or Statistical Analysis Plan (SAP). We provide more discussion on how to choose values for these parameters to meet specific study objectives in Section \ref{s:sim_power} and \ref{s:real}. The impact of design parameters are also evaluated later in Section \ref{s:sim_power}.  

\begin{figure}
	\centerline{\includegraphics[width=6.3in]{RAR_flowchart.png}}
	\caption{RABR randomization procedures.}
	\label{F:RAR_flowchart}
\end{figure}


\section{Inference Procedures}
\label{s:inf}

\subsection{Type I error rate in pairwise comparison for a continuous endpoint}
\label{s:con_error}

For demonstration purposes in this section, we assume that $\sigma_1 = \sigma_2 = 1$. The following proof can be generalized to situations with unknown heterogeneous variances by substituting $\sigma_1$ and $\sigma_2$ by their consistent estimators $\widehat{\sigma}_1$ and $\widehat{\sigma}_2$, respectively. To test the null hypothesis $H_0: \mu_0 = \mu_1$ against the one-sided alternative hypothesis $H_1: \mu_0 < \mu_1$, one can use the following {\it z} statistic $S$, which is defined by 
\begin{equation}
\label{equ_un_stats_all}
	S = \left\{ \xbar{X}^{(1)} - \xbar{X}^{(0)}\right\} \Bigg/ \left\{\frac{1}{N^{(1)}}+\frac{1}{N^{(0)}}\right\}^{1/2}. 
\end{equation}
In the two-stage setting, the unweighted statistic $S$ can be decomposed as $S\left(N_2^{(1)}, N_2^{(1)} \right)$,
\begin{align}
S = S\left(N_2^{(1)}, N_2^{(1)} \right) & = \left\{\frac{\xbar{X}_1^{(1)} N_1^{(1)} + \xbar{X}_2^{(1)} N_2^{(1)} }{ N_1^{(1)} +  N_2^{(1)}} - \xbar{X}^{(0)}\right\} \Bigg/ \left\{\frac{1}{N_1^{(1)}+N_2^{(1)}}+\frac{1}{N^{(0)}}\right\}^{1/2} \nonumber \\
& = \left[w_1\left(N_2^{(1)}\right)\right]^{1/2} Z_1 + \left[1-w_1\left(N_2^{(1)}\right)\right]^{1/2} Z_2\left(N_2^{(1)}\right) \label{stat_un} ,
\end{align}
where
\begin{align}
	S\left(a, b \right) = & \{w_1(b)\}^{1/2} Z_1 + \{1-w_1(b)\}^{1/2} Z_2(a) \label{step1_stats}, \\
	Z_1 = & \left\{N_1^{(1)}\right\}^{1/2} \: \left\{ \xbar{X}_1^{(1)}-\mu_0 \right\}, \nonumber \\
	Z_2(a) = & \left[\frac{a}{N_1^{(1)}+a}\left\{\xbar{X}_2^{(1)}-\mu_0 \right\} - \left\{\xbar{X}^{(0)}-\mu_0 \right\}\right]\Bigg/{\left[\frac{a}{\left\{N_1^{(1)}+a\right\}^2}+\frac{1}{N^{(0)}}\right]^{1/2}}, \nonumber  \\
	w_1(b) = &  {\frac{N_1^{(1)}}{\left\{N_1^{(1)}+b\right\}^2}}\Bigg/\left\{\frac{1}{N_1^{(1)}+b}+\frac{1}{{N^{(0)}}}\right\}, \label{w1}
\end{align}

One can show that $Z_1$ and $Z_2\left(N_2^{(1)}\right)$ follow two independent standard normal distributions $\mathcal{N}(0, 1)$ under the null hypothesis $H_0$. \citep{cui1999} Even though $N_2^{(1)}$ depends on $\boldsymbol{D}_1$, the fact that $Z_2\left(N_2^{(1)}\right)$ follows the same $\mathcal{N}(0, 1)$ given different realizations of $N_2^{(1)}$ (either $\beta_1 N_2^{(0)}$ or $\beta_2 N_2^{(0)}$) ensures the marginal independence between $Z_1$ and $Z_2\left(N_2^{(1)}\right)$. 

In a non-adaptive design, $N_2^{(1)}$ takes a fixed value. By the normal response assumption in (\ref{notation_x_normal}), $S\left(N_2^{(1)}, N_2^{(1)} \right)$ follows $\mathcal{N}(0, 1)$ under the null hypothesis $H_0$. Therefore, $\text{pr}_{H_0}\left[ S\left(N_2^{(1)}, N_2^{(1)} \right) > c \right] = \alpha$, where $c = {\Phi}^{-1}(1-\alpha)$ and $\Phi^{-1}(\cdot)$ is the inverse of the cumulative distribution function of a standard normal distribution. The probability of erroneously detecting a significant low-dose treatment effect with no multiplicity adjustment under $H_0$ is exactly equal to a nominal level $\alpha$, for example $\alpha = 2.5\%$ in practice. 

In adaptive design where $N_2^{(1)}$ depends on $\boldsymbol{D}_1$, $S\left(N_2^{(1)}, N_2^{(1)} \right)$ does not follow $\mathcal{N}(0, 1)$ because the weight $\left[w_1\left(N_2^{(1)}\right)\right]^{1/2}$ is a function of $N_2^{(1)}$. Proper adjustment is required to protect the type I error rate at $\alpha$. A popular approach is to use CHW statistics $S\left(N_2^{(1)}, N_2^{(0)} \right)$ with a pre-specified weight $\left[w_1\left(N_2^{(0)}\right)\right]^{1/2}$, where $\text{pr}_{H_0}\left[ S\left(N_2^{(1)}, b \right) > c \right] = \alpha$ as long as $w_1(b) \in (0, 1)$ is fixed, and is independent from interim data $\boldsymbol{D}_1$. \citep{cui1997, cui1999}

In our proposed RABR design, we still want to utilize the conventional statistics $S\left(N_2^{(1)}, N_2^{(1)} \right)$ in (\ref{step1_stats}) and the significance cutoff $c = \Phi^{-1}(1-\alpha)$ in our inference procedure. However, the control of type I error rate requires further investigation. Before studying the one-sided type I error control, we first provide the following Lemma,

\begin{lemma}
	
	For any constant $c^\prime \in \mathbb{R}$, $Q(w_1) = \text{pr} \left[(w_1)^{1/2} Z_1 + 
	(1-w_1)^{1/2} Z_2 > c \mid Z_1 \leq Z_3 + c^\prime \right]$ is a decreasing function with respect to $w_1 \in (0, 1)$ for $Z_1$, $Z_2$ and $Z_3$ following mutually independent $\mathcal{N}(0, 1)$.
\end{lemma}

The proof is provided in the Supplemental Materials Section 1. In the following theorem, we prove that the probability of falsely detecting treatment effect with the decision rule $S\left(N_2^{(1)}, N_2^{(1)} \right) > c$ in the proposed design under $H_0$ is upper bounded by $\alpha$.

\begin{theorem}
	\label{theorem_error}
	In the RABR design with a normal response in (\ref{notation_x_normal}) and two active treatment groups, we have
	\begin{equation}
		\label{proof_equ}
		{\text{pr}}_{H_0}\left[S\left(N_2^{(1)}, N_2^{(1)} \right) > c \right] \leq \alpha. 
	\end{equation} 
\end{theorem}

\begin{proof}
	The randomization vector ${r}$ satisfies $r_1 \geq r_2 \geq 0$, which is equivlent to $\beta_1 \geq \beta_2 \geq 0$ given $r_0$ at each stage. Consider
	\begin{align}
		& \left[w_1\left(\beta_1 N_2^{(0)}\right) - w_1\left(\beta_2 N_2^{(0)}\right) \right]\left\{\frac{1}{N_1^{(1)}+\beta_1 N_2^{(0)}}+\frac{1}{N^{(0)}} \right\}\left\{\frac{1}{N_1^{(1)}+\beta_2 N_2^{(0)}}+\frac{1}{N^{(0)}} \right\} \nonumber \\
		= & \frac{N_1^{(1)}}{\left\{N_1^{(1)} + \beta_1 N_2^{(0)}\right\}^2 \left\{N_1^{(1)} + \beta_2 N_2^{(0)}\right\}} - \frac{N_1^{(1)}}{\left\{N_1^{(1)} + \beta_2 N_2^{(0)}\right\}^2\left\{N_1^{(1)} + \beta_1 N_2^{(0)}\right\}} + \label{inequ_two_w1_1} \\
		& \frac{N_1^{(1)}}{\left\{N_1^{(1)} + \beta_1 N_2^{(0)}\right\}^2 N^{(0)}} - \frac{N_1^{(1)}}{\left\{N_1^{(1)} + \beta_2 N_2^{(0)}\right\}^2 N^{(0)}}. \label{inequ_two_w1_2} 
	\end{align}
	Given that $\beta_1 \geq \beta_2 \geq 0$, one can verify that (\ref{inequ_two_w1_1})$ \leq 0$, (\ref{inequ_two_w1_2})$ \leq 0$, and therefore $w_1\left(\beta_1 N_2^{(0)}\right) \leq w_1\left(\beta_2 N_2^{(0)}\right)$. 
	
	Based on the decomposition of $S\left(N_2^{(1)}, N_2^{(1)} \right)$ in (\ref{stat_un}), $Z_1$ and $Z_2\left(N_2^{(1)}\right)$ follow two independent $\mathcal{N}(0, 1)$'s under null $H_0$. We drop $\left(N_2^{(1)}\right)$ in the notation of $Z_2\left(N_2^{(1)}\right)$ for simplicity. Under $H_0$, we have
	\begin{align}
		& \text{pr}_{H_0}\left[ S\left(N_2^{(1)}, N_2^{(1)} \right) > c \right] \nonumber \\
		= & \text{pr}_{H_0} \left(\left[w_1\left(\beta_1 N_2^{(0)}\right)\right]^{1/2}Z_1 + \left[1-w_1\left(\beta_1 N_2^{(0)}\right)\right]^{1/2}Z_2 > c \mid Z_1 \geq Z_1^{(2)} + c^\prime \right) \text{pr}\left\{ Z_1 \geq Z_1^{(2)} + c^\prime \right\} + \nonumber \\
		& \text{pr}_{H_0} \left(\left[w_1\left(\beta_2 N_2^{(0)}\right)\right]^{1/2}Z_1 + \left[1-w_1\left(\beta_2 N_2^{(0)}\right)\right]^{1/2}Z_2 > c \mid Z_1 < Z_1^{(2)} + c^\prime \right) \text{pr}\left\{Z_1 < Z_1^{(2)} + c^\prime \right\}, \label{step3_prob_random}
	\end{align}
	where $Z_1^{(2)} = \left\{N_1^{(2)}\right\}^{1/2} \: \left\{\xbar{X}_1^{(2)} - \mu_0\right\}$, $c^\prime = \left\{N_1^{(2)}\right\}^{1/2} \: \mu_0 - \left\{N_1^{(1)}\right\}^{1/2} \: \mu_0$. Note that the condition $Z_1 \geq Z_1^{(2)} + c^\prime$ is equivalent to our arm selection criteria if treatment group $1$ is more promising, as defined by the evaluation function (4). On the other hand, consider a CHW statistic with a fixed weight $\left[w_1\left(\beta_1 N_2^{(0)}\right)\right]^{1/2}$ for example, we have an analytically type I error control at $\alpha$ :
	\begin{align}
		\alpha = & \text{pr}\left[ S\left(N_2^{(1)}, \beta_1 N_2^{(0)} \right) > c \right] \nonumber \\
		= & \text{pr} \left(\left[w_1\left(\beta_1 N_2^{(0)}\right)\right]^{1/2}Z_1 + \left[ 1-w_1\left(\beta_1 N_2^{(0)}\right)\right]^{1/2}Z_2 > c \right) \nonumber \\
		= & \text{pr} \left(\left[w_1\left(\beta_1 N_2^{(0)}\right)\right]^{1/2}Z_1 + \left[1-w_1\left(\beta_1 N_2^{(0)}\right)\right]^{1/2}Z_2 > c \mid Z_1 \geq Z_1^{(2)} + c^\prime \right) \text{pr}\left\{ Z_1 \geq Z_1^{(2)} + c^\prime \right\} + \nonumber \\
		& \text{pr} \left(\left[w_1\left(\beta_1 N_2^{(0)}\right)\right]^{1/2}Z_1 + \left[1-w_1\left(\beta_1 N_2^{(0)}\right)\right]^{1/2}Z_2 > c \mid Z_1 < Z_1^{(2)} + c^\prime \right) \text{pr}\left\{Z_1 < Z_1^{(2)} + c^\prime \right\}. \label{step3_prob_CHW}
	\end{align}
	From Lemma 1 and $w_1\left(\beta_1 N_2^{(0)}\right) \leq w_1\left(\beta_2 N_2^{(0)}\right)$ shown previously, we have 
	\begin{align*}
		& \text{pr} \left(\left[w_1\left(\beta_2 N_2^{(0)}\right)\right]^{1/2}Z_1 + \left[1-w_1\left(\beta_2 N_2^{(0)}\right)\right]^{1/2}Z_2 > c \mid Z_1 < Z_1^{(2)} + c^\prime \right) \nonumber \\
		\leq  & \text{pr} \left(\left[w_1\left(\beta_1 N_2^{(0)}\right)\right]^{1/2}Z_1 + \left[1-w_1\left(\beta_1 N_2^{(0)}\right)\right]^{1/2}Z_2 > c \mid Z_1 < Z_1^{(2)} + c^\prime \right), \label{step3_expand_1}
	\end{align*}
	and therefore (\ref{step3_prob_random}) $\leq$ (\ref{step3_prob_CHW}). Theorem 1 is proved. 
\end{proof}

Next we generalize Theorem \ref{theorem_error} from a single adaptation to one with multiple checkpoints in the following Theorem:
\begin{theorem}
	\label{theorem_error_multiple}
	In the RABR design with a normal response in (\ref{notation_x_normal}), two active treatment groups and more than two adaptation timepoints, we have
	\begin{equation}
		\label{proof_equ}
		{\text{pr}}_{H_0}\left[ S > c \right] \leq \alpha,
	\end{equation} 
where $S$ is the unweighted statistic defined in (\ref{equ_un_stats_all}). 
\end{theorem}

The proof follows by iteratively applying Theorem \ref{theorem_error} using backward induction from the last two stages to early stages. The idea is based on Brannath et al, \cite{bra2002} who prove that the type I error rate is equal to the nominal level $\alpha$ in an adaptive design with multiple stages by using fixed weights in (\ref{w1}). Theorem \ref{theorem_error} and \ref{theorem_error_multiple} can be straightforwardly generalized to accommodate a study with three active treatment groups, with proof in the Supplemental Materials Section 2. The generalization to $m>3$ active treatment groups needs to be cautious, as the analytic proof of type I error protection needs additional conditions on the study design. More discussion on this generalization is provided in Section \ref{s:dis}. In Section 3 of the Supplemental Materials, we illustrate a corresponding theorem of using the proportional test in the context of a binary endpoint.  

The error rate protection in pairwise comparison can be generalized to other intersection hypotheses as well. Therefore, by applying proper multiple testing procedures, for example the Bonferroni test based on the closure principle, \citep{bre2016} it is sufficient to protect the one-sided FWER at $\alpha$ in the strong sense.

\section{Simulation studies}
\label{s:sim}

By simulations, we study the type I error rate of RABR design with a continuous endpoint at Section \ref{s:sim_error}. The final allocations and power are further compared with DBCD at Section \ref{s:sim_power} under several response assumptions. 

Consider a clinical trial with three active treatment dosing groups, $D_1$, $D_2$, $D_3$, and a placebo. There are $M=60$ subjects being recruited in the burn-in period with fixed equal randomization probabilities and a subsequent $60$ subjects being adaptively and sequentially randomized. Therefore, the total sample size is $N=120$. We utilize the step-down Dunnett test to control multiplicity. \citep{bre2016} At the final efficacy analysis stage, the dosing group with the smallest step-down Dunnett adjusted $p$-value is chosen as the selected one. Its efficacy is further confirmed if its adjusted $p$-value is smaller than $\alpha$. We use $S_1$ to denote the selected treatment group, and use $S_2$ and $S_3$ for the second best and worst groups.

The study objective is to identify and confirm the efficacy of the best performing dose $S_1$ with FWER controlled at one-sided $\alpha = 2.5\%$, while achieving a desired final sample size for that group and the placebo. Note that the choice of $S_1$ from $D_1$, $D_2$ and $D_3$ depends on the data, and does not necessarily correspond to a particular dosing group.  

\subsection{Type I error control}
\label{s:sim_error}

To evaluate the type I error rate of using the usual Student's $t$ statistics in RABR, we set the response mean for all groups to be the same, that is $\mu_g = \widetilde{\mu}$ for $g = 0, 1, 2, 3$, where $\widetilde{\mu}$ is considered at $0$ or $1$. The variance $\sigma^2$ is assumed to be $1$ for all groups. The following $5$ different adaptive randomization ratio vectors $\boldsymbol{r}$'s are considered:

\begin{enumerate}
	\item $\boldsymbol{r}_a =  (8, 4, 4, 4) $,
	\item $\boldsymbol{r}_b =  (8, 5, 4, 3) $,
	\item $\boldsymbol{r}_c =  (8, 7, 4, 1) $,
	\item $\boldsymbol{r}_d =  (8, 5, 5, 2) $,
	\item $\boldsymbol{r}_e =  (9, 9, 1, 1) $.
\end{enumerate}

For example in $\boldsymbol{r}_b$, the randomization probability to placebo is fixed at $8$ out of $20$ in the RAR period, while the best performing dose gets $5$ out of $20$. The first $\boldsymbol{r}_a$ corresponds to a non-adaptive design with fixed equal randomization probability to each treatment group. We use this benchmark to demonstrate that the type I error rates of RABR in $\boldsymbol{r}_b$, $\boldsymbol{r}_c$, $\boldsymbol{r}_d$ and $\boldsymbol{r}_e$ are controlled at $\alpha = 2.5\%$ as proved in Theorem \ref{theorem_error} and \ref{theorem_error_multiple}. The number of simulation iterations is $100,000$. We also consider a small size trial with a total of $n=40$ subjects and burn-in size of $M=20$. 

As can be seen from Table \ref{T:error_con}, the probability of erroneously claiming a significant treatment effect with no multiplicity adjustment in each dose is controlled under $\alpha = 2.5\%$ for RABR with $\boldsymbol{r}_b$, $\boldsymbol{r}_c$, $\boldsymbol{r}_d$ and $\boldsymbol{r}_e$, while it is around $2.5\%$ under the non-adaptive design $\boldsymbol{r}_a$. The error rate in pairwise comparison is smaller than the nominal level, as we proved in Theorem \ref{theorem_error} and \ref{theorem_error_multiple}. With the step-down Dunnett procedure to adjust multiplicity, the adjusted error rate for each dose is much smaller than $\alpha$, while the overall FWER is controlled at $\alpha$ as well.

\begin{table}[ht]
	\centering
	\caption{The unweighted Student's $t$ test statistics have controlled one-sided type I error rates at the $2.5\%$ level in RABR design with a continuous endpoint.}
	\label{T:error_con}
	\begin{tabular}{cccccccccc}
		\toprule
		$N$ &  $\mu_g = \widetilde{\mu}$ & $\boldsymbol{r}$ & \multicolumn{7}{c}{Probability of rejecting null hypothesis} \\ 
		\cmidrule(lr){4-10}
		& for $g=0, 1, 2, 3$&  & \multicolumn{3}{c}{with no multiplicity adjustment} & \multicolumn{4}{c}{with step-down Dunnett} \\ 
		\cmidrule(lr){4-6}\cmidrule(lr){7-10}
		&  & & $D_1$ & $D_2$ & $D_3$ & $D_1$ & $D_2$ & $D_3$ & overall \\ 
	\midrule
		\\
120 & 0 & $\boldsymbol{r}_a=($8, 4, 4, 4$)$ & 2.45\% & 2.49\% & 2.48\% & 0.93\% & 0.96\% & 0.96\% & 2.46\% \\
&  &  $\boldsymbol{r}_b=($8, 5, 4, 3$)$ & 2.38\% & 2.36\% & 2.35\% & 0.91\% & 0.92\% & 0.92\% & 2.39\% \\
&  &  $\boldsymbol{r}_c=($8, 7, 4, 1$)$ & 2.21\% & 2.21\% & 2.18\% & 0.87\% & 0.85\% & 0.85\% & 2.27\% \\
&  &  $\boldsymbol{r}_d=($8, 5, 5, 2$)$ & 2.40\% & 2.36\% & 2.46\% & 0.93\% & 0.94\% & 0.95\% & 2.44\% \\
&  &  $\boldsymbol{r}_e=($9, 9, 1, 1$)$ & 2.01\% & 1.98\% & 1.93\% & 0.81\% & 0.76\% & 0.74\% & 2.13\% \\
\\
& 1 &  $\boldsymbol{r}_a=($8, 4, 4, 4$)$ & 2.46\% & 2.52\% & 2.56\% & 0.99\% & 0.99\% & 1.01\% & 2.55\% \\
&  &  $\boldsymbol{r}_b=($8, 5, 4, 3$)$ & 2.41\% & 2.42\% & 2.44\% & 1.00\% & 0.92\% & 0.96\% & 2.47\% \\
&  &  $\boldsymbol{r}_c=($8, 7, 4, 1$)$ & 2.17\% & 2.17\% & 2.21\% & 0.86\% & 0.85\% & 0.89\% & 2.25\% \\
&  &  $\boldsymbol{r}_d=($8, 5, 5, 2$)$ & 2.44\% & 2.49\% & 2.45\% & 0.98\% & 0.94\% & 0.94\% & 2.46\% \\
&  &  $\boldsymbol{r}_e=($9, 9, 1, 1$)$ & 1.99\% & 1.91\% & 1.96\% & 0.81\% & 0.71\% & 0.75\% & 2.01\% \\
\\
40 & 0 &  $\boldsymbol{r}_a=($8, 4, 4, 4$)$ & 2.53\% & 2.44\% & 2.47\% & 1.03\% & 0.96\% & 0.97\% & 2.48\% \\
&  &  $\boldsymbol{r}_b=($8, 5, 4, 3$)$ & 2.45\% & 2.40\% & 2.29\% & 0.98\% & 0.93\% & 0.92\% & 2.40\% \\
&  &  $\boldsymbol{r}_c=($8, 7, 4, 1$)$ & 2.11\% & 2.14\% & 2.16\% & 0.83\% & 0.83\% & 0.87\% & 2.19\% \\
&  &  $\boldsymbol{r}_d=($8, 5, 5, 2$)$ & 2.50\% & 2.39\% & 2.38\% & 0.99\% & 0.96\% & 0.95\% & 2.48\% \\
&  &  $\boldsymbol{r}_e=($9, 9, 1, 1$)$ & 1.88\% & 1.93\% & 1.93\% & 0.72\% & 0.75\% & 0.75\% & 2.04\% \\
\\
& 1 &  $\boldsymbol{r}_a=($8, 4, 4, 4$)$ & 2.49\% & 2.53\% & 2.49\% & 1.00\% & 0.99\% & 1.00\% & 2.53\% \\
&  &  $\boldsymbol{r}_b=($8, 5, 4, 3$)$ & 2.38\% & 2.39\% & 2.35\% & 0.94\% & 0.95\% & 0.94\% & 2.41\% \\
&  &  $\boldsymbol{r}_c=($8, 7, 4, 1$)$ & 2.13\% & 2.24\% & 2.16\% & 0.84\% & 0.91\% & 0.89\% & 2.29\% \\
&  &  $\boldsymbol{r}_d=($8, 5, 5, 2$)$ & 2.45\% & 2.39\% & 2.38\% & 0.95\% & 0.94\% & 0.97\% & 2.46\% \\
&  &  $\boldsymbol{r}_e=($9, 9, 1, 1$)$ & 1.83\% & 1.86\% & 1.88\% & 0.75\% & 0.80\% & 0.78\% & 2.03\% \\
		\bottomrule
	\end{tabular}
\end{table}

\subsection{Power and final allocations}
\label{s:sim_power}

Next we evaluate the power and final allocations of our RABR against DBCD under several alternative hypotheses. Given a total sample size of $N=120$, let's consider a study with a target sample size for the placebo and the selected dose at $42$, which corresponds to $35\%$ of the total sample size.  

In RABR, supposing that we choose $M=60$ in the burn-in period, one can utilize the adaptive randomization vector $\boldsymbol{r}_e =  (9, 9, 1, 1) $ to achieve the target sample size of $42$ for the selected group and placebo. Note that in order to exactly reach the desired size, the best performing group identified based on the burn-in period should consistently receive a $9$ out of $20$ randomization probability throughout the RAR process. As shown later on,  our RABR leans toward the desired allocations because the probability of flipping is relatively small. Discussion on how to chose design parameters of RABR is provided later. 

When it comes to DBCD, however, one needs to assume a certain response vector $\boldsymbol{\mu} = \left(\mu_0, \mu_1, \mu_2, \mu_3\right)$, and fine tune the parameter $\lambda$ in (\ref{y_g_original}) to meet the study objective that both the placebo and the selected arm have a $35\%$ final sample size proportion. Since the response mean in the placebo is usually different from that in treatment groups, equation (\ref{y_g_original}) requires further adjustment. Even with a modified allocation function, the choice of the tuning parameter $\lambda$ is sensitive to the assumed $\boldsymbol{\mu}$ at the study design stage, which can lead to undesired allocations if observed data deviate. The function form in (\ref{y_g_original}) is used for demonstration. 

Consider the following three assumptions on response means from the placebo and three doses,
\begin{enumerate}
	\item $\boldsymbol{\mu}_A = (0.43, 0.48, 0.63, 1.2)$,	
	\item $\boldsymbol{\mu}_B = (0.43, 0.68, 0.93, 1.2)$, 
	\item $\boldsymbol{\mu}_C = (0.43, 1, 1.15, 1.2)$.
\end{enumerate}
Table \ref{T:power_con} summarizes the multiplicity adjusted power of the selected arm in RABR, a non-adaptive design with equal randomization probabilities, and DBCD with $\lambda = -2, 0, 2$. The overall multiplicity adjusted power is the probability of rejecting at least one elementary null hypothesis, or equivalently rejecting the null hypothesis in the best performing group $S_1$. For each individual group $D_1$, $D_2$ and $D_3$, we report the probability of selecting this group (if its adjusted $p$-value is the smallest among three doses) and confirming its efficacy (if its adjusted $p$-value is smaller than $\alpha$). Therefore, the sum of them from three groups is equal to the overall power. Our RABR has the largest overall power gain in $\boldsymbol{\mu}_A$ compared to a non-adaptive design ($83.27\%$ versus $72.32\%$) and DBCD with $\lambda = 2$ ($83.27\%$ versus $75.38\%$). It also has the highest probability of identifying the right dose $D_3$ and confirming its efficacy at $82.35\%$, compared to $71.72\%$ in a non-adaptive design and $75.06\%$ in DBCD with $\lambda = 2$. In $\boldsymbol{\mu}_B$ and $\boldsymbol{\mu}_C$, our method also has moderate overall power gain of claiming significance in the right treatment group (Table \ref{T:power_con}). 

Average sample sizes (ASNs) for selected treatment groups are presented at Table \ref{T:asn_con}. Across all three underlying response $\boldsymbol{\mu}$'s, RABR generally approaches the final desired sample size for the placebo, $S_1$, $S_2$ and $S_3$ at $42$, $42$, $18$ and $18$, respectively. The ASNs ($41.99$ in placebo and $40.44$ in $S_1$) are closer to targets under $\boldsymbol{\mu}_A$ where the best dose is separated well from the other two doses, but deviate moderately under $\boldsymbol{\mu}_C$, where response means of active treatment groups are close at $1$, $1.15$ and $1.2$. On the other hand, the ASNs of the placebo group and the selected group $S_1$ in DBCD do not reach the target of $42$. Figure \ref{F:sample_size} visualizes the sample proportion of RABR and DBCD with $\lambda=2$ at each of the $60$ adaptations after burn-in period. The final observed proportions are reported in the text to the right of each sub-plot, while the targets are projected by the horizontal dashed lines. Starting from an equal sample proportion of $25\%$ right after the burn-in period, our RABR approaches the desired treatment group allocations.

We provide some remarks on the superior power performance and the robustness of reaching target sample size in RABR as compared with DBCD. As shown in Table 2, RABR has much higher power than DBCD with varying $\lambda$ values. The reason is that RABR can effectively increase the sample size of placebo and $S_1$ to the target of 42 but decrease those in the two worst performing groups to $18$. For example under $\boldsymbol{\mu_A}$, RABR achieves $41.99$ in placebo and $40.44$ in $S_1$, while DBCD with $\lambda = -2$ has $29.96$ for placebo and $30.07$ for $S_1$. Moreover, RABR reaches the target sample size robustly under different treatment assumptions $\boldsymbol{\mu_A}$, $\boldsymbol{\mu_B}$ and $\boldsymbol{\mu_C}$ (Table \ref{T:asn_con}). On the other hand, the ASNs of DBCD are sensitive to the choice of parameter $\lambda$, which is selected at the study design stage with assumptions. When observed data deviate, our proposed RABR has a more robust performance than DBCD in reaching target sample size. 

Next, we discuss how to choose design parameters of RABR to meet a specific requirement of final sample size based on additional analysis in Table \ref{T:power_con_sen}. With the burn-in size $M$, a small value can lead to unstable estimation of treatment effect and may allow RABR to assign more subjects to inferior arms in early stages, while a large value does not leave much room for later adaptations. We suggest proposing several candidate values of $M$ around $50\%$ of total sample size $N$, and choosing a proper one based on simulations. Taking the $5$ rows under $\boldsymbol{\mu}_A$ as an example, we first evaluate three magnitudes of $M$ at $40$, $60$ and $24$ on the performance of RABR. Their corresponding $\boldsymbol{r}$'s with the block size around $20$ are chosen based on the study objective of reaching $42$ in both placebo and $S_1$. With a moderate $M=40$, RABR has a slightly higher power of selecting and confirming efficacy of $D_3$ at $82.37\%$ as compared with $M=60$ and $M=24$. In terms of ASN, $M=60$ performs better by having $40.41$ in $S_1$ because the estimates in (\ref{notation_measure}) for early adaptations are more accurate with a larger $M$. Under the same $M=40$, the next two rows evaluate $\boldsymbol{r} = (16, 16, 7, 1)$ with a larger batch size $B$, and $\boldsymbol{r} = (4, 4, 1, 1)$ with a smaller $B$. As mentioned previously, since each subject is sequentially randomized after burn-in period, then these three values of $\boldsymbol{r}$ with the same $M=40$ have the same number of adaptations at $80$. In general, $\boldsymbol{r} = (16, 16, 7, 1)$ has higher power and more accurate sample size than $\boldsymbol{r} = (4, 4, 1, 1)$. The reason is that if the best group $D_3$ has the second-best performance in early adaptations, it still has a relatively high probability of $7$ out of $40$ to get patients assigned and converge to the underling response mean later on. The above observations are generally consistent under $\boldsymbol{\mu}_B$ and $\boldsymbol{\mu}_C$ as well. In practice, one may choose $M=40$ with $\boldsymbol{r} = (16, 16, 7, 1)$ to reach a higher power or choose $M=60$ with $\boldsymbol{r} = (9, 9, 1, 1)$ with a more accurate sample size. 

\subsection{Sample size and power approximation}

In this section, we provide some guidance on how to approximate the sample size and power of RABR. Note that in order to fully access the power performance and operating characteristics of RABR, one needs to conduct simulations under varying design parameters and varying underlying responses. The below approximation is a starting point to obtain a rough estimate of power or sample size. 

We first discuss about how to approximate power given sample size using the hypothetical example in Section \ref{s:sim_power}. As a starting point, we can use two-sample $t$ test with Bonferroni correction to approximate the overall power in Table \ref{T:power_con}, which is the probability of claiming efficacy in at least one treatment group with multiplicity adjustment. The sample size of this two-sample $t$ test is $42$ for each group, which is our target sample size of placebo and $S_1$. The response means of two groups are set at $0.43$ from placebo, and $1.2$ from high dose in previous assumptions $\boldsymbol{\mu}_A$, $\boldsymbol{\mu}_B$, $\boldsymbol{\mu}_C$. The rationale is to consider a best case scenario where the high dose reaches the target final sample size of $42$. The one-sided significance level with Bonferroni correction can be easily calculated at $2.5\%$ divided by $3$, because we have $m=3$ doses. With standard software, we calculate the power at approximately $85\%$, which provides a pretty good estimate of the overall power under three underlying response mean vectors in Table \ref{T:power_con}. However, in order to further determine the probability of selecting and confirming the efficacy in each dose, we need to conduct simulations to implement RABR under varying responses. Simulations are also needed to choose design parameters in RABR as discussed previously. The link to our R code is provided in the Section of ``Supplementary Material''. 

To determine the sample size given a target power, one can also use the two-sample $t$ test with Bonferroni adjustment to approximate the sample size in selected groups given overall power, and then conduct simulations of RABR to further evaluate operating characteristics. 

\begin{table}[ht]
	\centering
	\caption{RABR has a higher power of selecting and confirming the efficacy of the best performing group ($D_3$) than the design with fixed randomization, and DBCD.}
	\label{T:power_con}
	\begin{tabular}{ccccccc}
		\toprule
		$\boldsymbol{\mu}$ & Method & $\lambda$ & \multicolumn{4}{c}{Step-down Dunnett adjusted power}  \\
		&  &  & \multicolumn{4}{c}{of selecting and confirming the efficacy in}   \\
		\cmidrule(lr){4-7}
		& & & $D_1$ & $D_2$ & $D_3$ & overall  \\ 
		\midrule
$\boldsymbol{\mu}_A$ & RABR & - & \bf 0.12\% & \bf 0.79\% & \bf 82.35\% & \bf 83.27\% \\ 
& Fixed & - & 0.07\% & 0.52\% & 71.72\% & 72.32\% \\ 
& DBCD & -2 & 0.08\% & 0.52\% & 71.99\% & 72.59\% \\ 
&  & 0 & 0.07\% & 0.45\% & 73.25\% & 73.77\% \\ 
&  & 2 & 0.04\% & 0.29\% & 75.06\% & 75.38\% \\ 
\\
$\boldsymbol{\mu}_B$ & RABR & - & \bf 1.08\% & \bf 12.24\% & \bf 69.11\% & \bf 82.44\% \\ 
& Fixed & - & 0.72\% & 9.46\% & 65.38\% & 75.56\% \\ 
& DBCD & -2 & 0.70\% & 9.52\% & 65.53\% & 75.76\% \\ 
&  & 0 & 0.62\% & 9.04\% & 66.23\% & 75.89\% \\ 
&  & 2 & 0.41\% & 7.51\% & 66.20\% & 74.13\% \\ 
\\
$\boldsymbol{\mu}_C$ & RABR & - & \bf 11.87\% & \bf 33.31\% & \bf 44.85\% & \bf 90.03\% \\ 
& Fixed & - & 9.57\% & 31.36\% & 44.27\% & 85.20\% \\ 
& DBCD & -2 & 9.57\% & 31.62\% & 44.33\% & 85.52\% \\ 
&  & 0 & 9.45\% & 31.21\% & 44.29\% & 84.95\% \\ 
&  & 2 & 7.94\% & 29.95\% & 43.10\% & 80.99\% \\ 
		\bottomrule
	\end{tabular}
\end{table}

\begin{table}[ht]
	\centering
	\caption{RABR reaches the target final sample size more robustly than the design with fixed randomization, and DBCD.}
	\label{T:asn_con}
	\begin{tabular}{ccccccc}
		\toprule
		$\boldsymbol{\mu}$ & Method & $\lambda$ & \multicolumn{4}{c}{Average sample size}  \\
		\cmidrule(lr){4-7}
		& & & placebo & $S_1$ & $S_2$ & $S_3$  \\ 
		\midrule
		\multicolumn{3}{r}{Target sample size: } & 42 & 42 & 18 & 18 \\ 
		\\
$\boldsymbol{\mu}_A$ & RABR & - & \bf 41.99 & \bf 40.44 & \bf 19.31 & \bf 18.27 \\ 
& Fixed & - & 30.00 & 30.03 & 30.02 & 29.96 \\ 
& DBCD & -2 & 29.96 & 30.07 & 30.02 & 29.95 \\ 
&  & 0 & 28.63 & 32.51 & 30.15 & 28.71 \\ 
&  & 2 & 24.04 & 43.36 & 28.74 & 23.86 \\ 
\\
$\boldsymbol{\mu}_B$ & RABR & - & \bf 42.00 & \bf 38.78 & \bf 20.67 & \bf 18.55 \\ 
& Fixed & - & 29.99 & 30.15 & 29.96 & 29.91 \\ 
& DBCD & -2 & 29.93 & 30.10 & 30.01 & 29.95 \\ 
&  & 0 & 28.00 & 31.86 & 30.76 & 29.39 \\ 
&  & 2 & 22.43 & 39.67 & 32.09 & 25.81 \\ 
\\
$\boldsymbol{\mu}_C$ & RABR & - & \bf 42.01 & \bf 38.00 & \bf 21.24 & \bf 18.75 \\ 
& Fixed & - & 29.99 & 30.33 & 30.01 & 29.67 \\ 
& DBCD & -2 & 29.91 & 30.16 & 30.04 & 29.89 \\ 
&  & 0 & 27.48 & 31.52 & 30.90 & 30.10 \\ 
&  & 2 & 20.91 & 37.26 & 33.21 & 28.62 \\ 
		\bottomrule
	\end{tabular}
\end{table}

\begin{figure}
	\centerline{\includegraphics[width=6.3in]{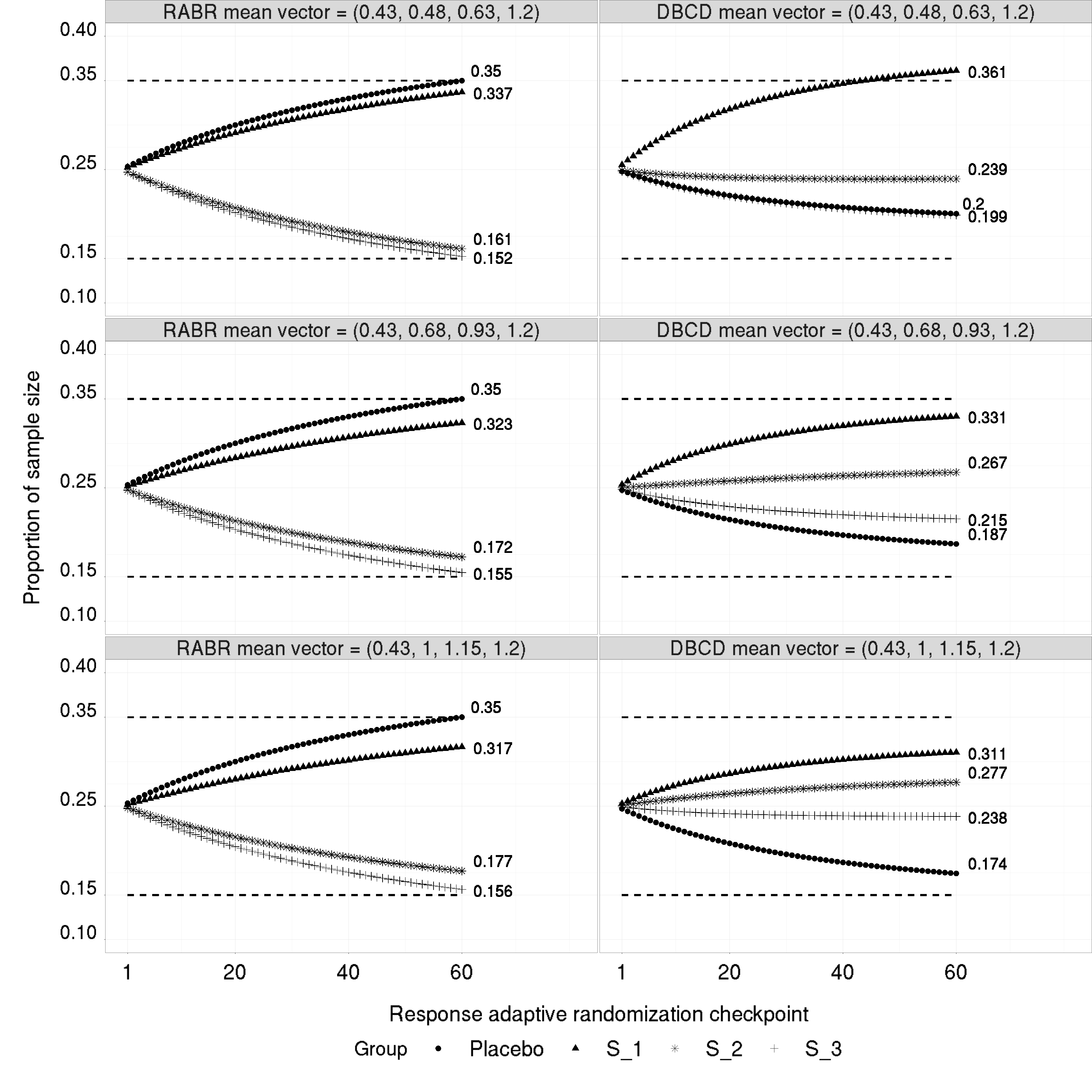}}
	\caption{RABR leans toward the target final sample size more robustly than DBCD with $\lambda=2$.}
	\label{F:sample_size}
\end{figure}

\begin{table}[ht]
\centering
\caption{Sensitivity analysis of RABR on the choices of randomization vector $\boldsymbol{r}$ and burn-in size $M$.}
\label{T:power_con_sen}
\begin{tabular}{ccccccccccc}
	\toprule
	$\boldsymbol{\mu}$ & $\boldsymbol{r}$ & $M$ & \multicolumn{4}{c}{Step-down Dunnett adjusted power} & \multicolumn{4}{c}{Average sample size} \\
	&  &  & \multicolumn{4}{c}{of selecting and confirming the efficacy in}   \\
	\cmidrule(lr){4-7} \cmidrule(lr){8-11}
	& & & $D_1$ & $D_2$ & $D_3$ & overall & placebo & $S_1$ & $S_2$ & $S_3$ \\ 
	\midrule
	 & & & & & \multicolumn{2}{c}{Target sample size: } & 42 & 42 & 18 & 18 \\ 
	\\
$\boldsymbol{\mu}_A$ & (8, 8, 3, 1) &    40 & 0.09\% & 0.72\% & 82.37\% & 83.17\% & 42.00 & 40.12 & 21.92 & 15.96 \\ 
& (9, 9, 1, 1) &    60 & 0.12\% & 0.79\% & 82.06\% & 82.97\% & 42.01 & 40.41 & 19.31 & 18.27 \\ 
& (9, 9, 5, 1) &    24 & 0.09\% & 0.69\% & 81.92\% & 82.70\% & 41.99 & 39.92 & 24.27 & 13.83 \\ 
& (16, 16, 7, 1) &    40 & 0.08\% & 0.66\% & 82.66\% & 83.39\% & 41.99 & 40.34 & 23.09 & 14.58 \\ 
& (4, 4, 1, 1) &    40 & 0.12\% & 0.83\% & 81.33\% & 82.28\% & 42.00 & 39.75 & 19.85 & 18.40 \\ 
\\
$\boldsymbol{\mu}_B$ & (8, 8, 3, 1) &    40 & 0.92\% & 12.00\% & 70.08\% & 83.00\% & 42.00 & 38.42 & 23.60 & 15.99 \\ 
& (9, 9, 1, 1) &    60 & 1.04\% & 12.62\% & 68.88\% & 82.55\% & 42.02 & 38.80 & 20.64 & 18.54 \\ 
& (9, 9, 5, 1) &    24 & 0.97\% & 11.76\% & 70.21\% & 82.94\% & 41.99 & 38.47 & 26.28 & 13.26 \\ 
& (16, 16, 7, 1) &    40 & 0.93\% & 11.74\% & 70.66\% & 83.32\% & 41.99 & 38.73 & 24.96 & 14.32 \\ 
& (4, 4, 1, 1) &    40 & 1.04\% & 12.59\% & 68.43\% & 82.06\% & 42.00 & 38.05 & 21.25 & 18.70 \\ 
\\
$\boldsymbol{\mu}_C$ & (8, 8, 3, 1) &    40 & 11.59\% & 33.44\% & 45.30\% & 90.33\% & 41.99 & 37.67 & 24.16 & 16.18 \\ 
& (9, 9, 1, 1) &    60 & 11.83\% & 33.74\% & 44.55\% & 90.13\% & 42.01 & 38.02 & 21.24 & 18.73 \\ 
& (9, 9, 5, 1) &    24 & 11.67\% & 33.84\% & 45.23\% & 90.74\% & 42.00 & 37.97 & 27.36 & 12.68 \\ 
& (16, 16, 7, 1) &    40 & 11.42\% & 33.64\% & 45.60\% & 90.66\% & 41.99 & 38.01 & 25.72 & 14.28 \\ 
& (4, 4, 1, 1) &    40 & 11.69\% & 33.67\% & 44.67\% & 90.03\% & 41.99 & 37.27 & 21.90 & 18.83 \\ 
	\bottomrule
\end{tabular}
\end{table}

\section{A case study}
\label{s:real}

Examining Everolimus in a Study of Tuberous Sclerosis Complex (EXIST-3) evaluated two dosing regimens of adjunctive everolimus compared with placebo for treatment-resistant focal-onset seizures in tuberous sclerosis complex. \citep{fre2016} In this phase 3, randomised, double-blinded, placebo-controlled study, eligible patients were equally and randomly assigned (1:1:1) to receive placebo, 3-7 ng/mL everolimus (low exposure) and 9-15 ng/mL (high exposure) everolimus with block randomization (block size of six). The primary endpoint was the proportion of patients achieving at least $50\%$ reduction in seizure frequency during a 12-week maintenance period, and therefore a larger proportion corresponds to a better outcome. A Bonferroni procedure was used to ensure an one-sided FWER at $2.5\%$. 

In the absence of previous dose-finding studies in this indication, the sample size of this phase 3 trial was determined to provide at least $90\%$ power to detect a difference in response rate from $15\%$ on placebo to $35\%$ in each of the two everolimus treatment groups based on literature review. \citep{fre2016} With limited prior knowledge to distinguish the efficacy of two treatment groups, equal randomization is a reasonable choice for this type of non-adaptive confirmatory trial. After completing this study, the results revealed that the response rate was $15.1\%$ with placebo compared with $28.2\%$ ($p=0.0077$) for low exposure everolimus and $40.0\%$ ($p<0.0001$) for high-exposure by using Cochran-Mantel-Haenszel Chi-square tests stratified by four age subgroups. More subjects could have been recruited to the high-exposure group due to a better efficacy profile and the fact that two doses have similar safety profiles. \citep{fre2016} In the remainder of this section, we re-design this non-adaptive design by our proposed RABR to assign more subjects to the better performing treatment group, and to reach the target sample size. 

For demonstration purposes, we simulate the trial with instant response based on the above underlying response rates with no stratification factor and a total sample size of $180$. Following notations in Section \ref{s:sim}, we use $D_1$ to denote the low exposure group, $D_2$ for the high exposure group. The selected treatment group $S_1$ is the arm between $D_1$ and $D_2$ with a smaller $p$-value in the test of proportion at the final analysis stage, and $S_2$ is the other arm as the worse performing treatment group. Suppose that the target sample size for the placebo and the selected arm $S_1$ is $72$. With a burn-in size $M=90$ with equal randomization, our RABR targets the desired allocations using an adaptive randomization vector $\boldsymbol{r} = (7, 7, 1)$. In DBCD, we consider a generalization of the Neyman allocation \citep{hu2006} as the target allocation function, 
\begin{equation*}
	\frac{\sqrt{q^{(g)}\left\{1-q^{(g)}\right\}}}{\sum_{j=0}^2 \sqrt{q^{(j)}\left\{1-q^{(j)}\right\}}},
\end{equation*}
where $q^{(g)}$ denotes the response rate in group $g$, and $g=0$ corresponds to placebo, $g=1$ for $D_1$, and $g=2$ for $D_2$. In all three methods considered (fixed randomization design, RABR and DBCD), the test of proportion is conducted to compare each exposure group versus placebo. The Bonferroni procedure is used to adjust for multiplicity. 

We first study the type I error rate when using the unweighted proportion test under our RABR design. Based on simulations with $100,000$ iterations, the one-sided error rates for two pairwise comparisons are lower than $\alpha =2.5\%$ under a comprehensive scan of the null response rate. The conservativeness of type I error rate is the consistent with the findings of continuous endpoints in Table \ref{T:error_con}, and is proved in Supplemental Materials Section 3. The probability of rejecting at least one elementary null hypothesis by the Bonferroni procedure is also upper bounded by $\alpha$ (Figure \ref{F:bin_error}).

\begin{figure}
	\centerline{\includegraphics[width=6.3in]{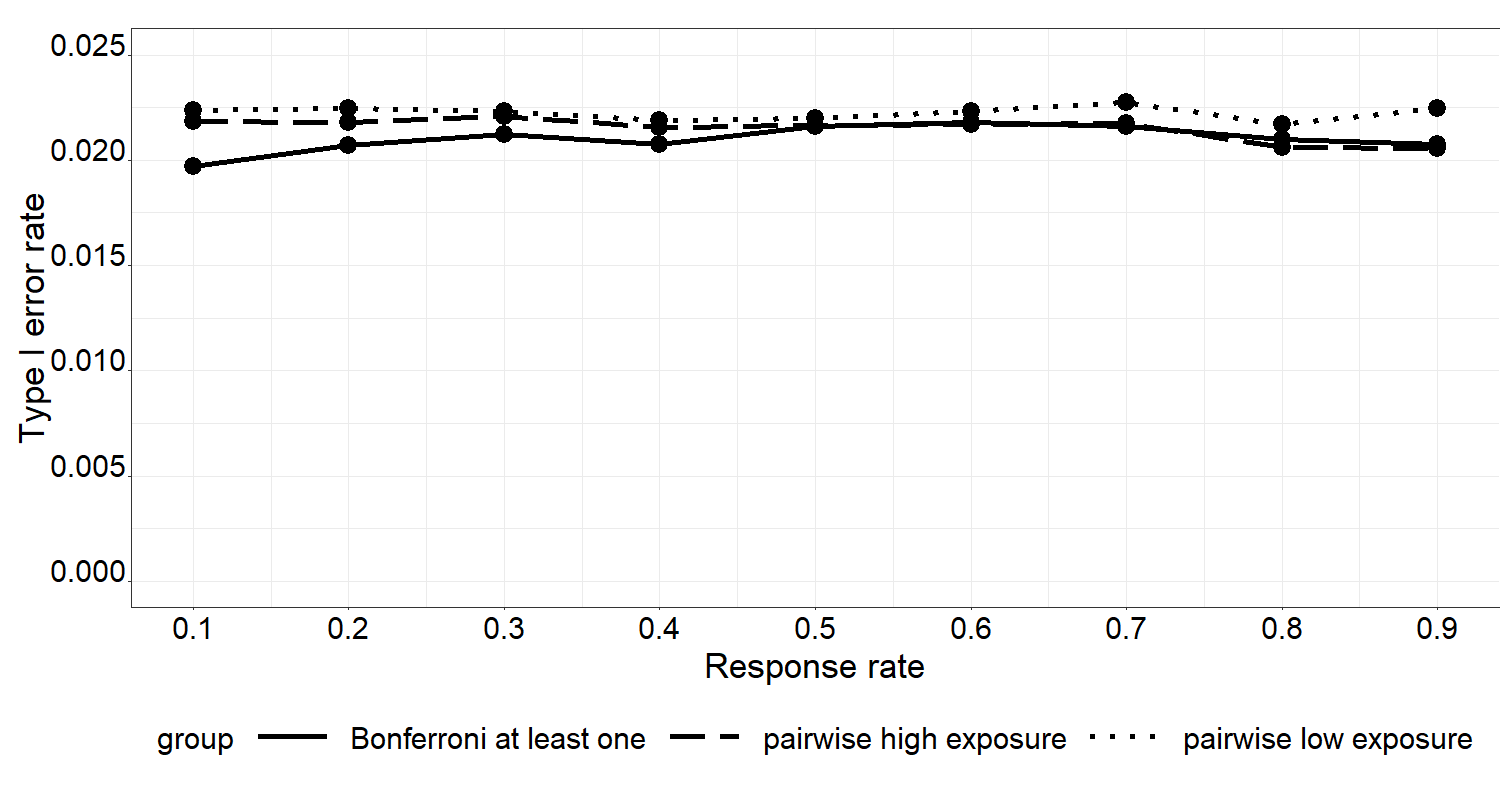}}
	\caption{RABR controls one-sided type I error rates not exceeding $2.5\%$ in pairwise comparisons and in claiming statistical significance for at least one dosing group with Bonferroni adjustment.}
	\label{F:bin_error}
\end{figure}

Next we evaluate power and operating characteristics of the three methods in Table \ref{T:bin_power}. Our RABR has the highest overall power at $86.22\%$, compared to the non-adaptive design with fixed randomization at $82.57\%$ and DBCD at $81.97\%$ (Table \ref{T:bin_power}). The probability of selecting and confirming the efficacy of the right dose ($D_2$; high exposure group) is also the highest ($78.40\%$). Another objective for two adaptive designs is to reach a target sample size of $72$ for both the placebo and the selected arm $S_1$. RABR achieves an average sample size of $72.02$ for the placebo and $69.93$ for $S_1$, which are closer to the target sample size than DBCD with $50.33$ for the placebo and $67.69$ for $S_1$. We observe that in general RABR has the largest variation in average sample size while fixed randomization has the smallest one based on standard deviation reported in parenthesis. 

\begin{table}[ht]
	\centering
	\caption{RABR achieves a higher power of selecting and confirming the efficacy of $D_2$, and reaches the target sample size more accurately than the design with fixed randomization and DBCD.}
	\label{T:bin_power}
	\begin{tabular}{lcccccc}
		\toprule
		Method & \multicolumn{3}{c}{Bonferroni adjusted power} & \multicolumn{3}{c}{Final sample size} \\ 
		& \multicolumn{3}{c}{of selecting and confirming} & \multicolumn{3}{c}{Mean (SD)} \\ 
		\cmidrule(lr){2-4} \cmidrule(lr){5-7} 
		& $D_1$ & $D_2$ & Overall & Placebo & $S_1$ & $S_2$  \\
		\midrule
RABR & \bf 7.82\% & \bf 78.40\% & \bf 86.22\% & \bf 72.02 (4.73) & \bf 69.93 (9.24) & \bf 38.05 (8.43) \\ 
Fixed randomization & 5.83\% & 76.75\% & 82.57\% & 60.02 (4.49) & 60.08 (4.47) & 59.91 (4.49) \\ 
DBCD & 4.87\% & 77.10\% & 81.97\% & 50.33 (5.88) & 67.69 (4.35) & 61.98 (4.63) \\ 
		\bottomrule
	\end{tabular}
\end{table}

\section{Concluding remarks}
\label{s:dis}

In this manuscript, we propose a practical Response Adaptive Block Randomization (RABR) design to adaptively assign more subjects to promising treatment groups based on accumulating interim data for multi-arm studies. Simulation studies show that our RABR robustly and accurately achieves the target allocations for the placebo and the selected treatment groups under varying underlying responses. This property, which is usually required in industry-sponsored clinical trials, makes our RABR more appealing in practice. The built-in block randomization feature eases the implementation of randomization procedures based on the IRT system. Moreover, we prove that the one-sided type I error rate from pairwise comparison of using conventional unweighted statistics in RABR is analytically controlled at a nominal level $\alpha$, which facilitates its application in confirmatory clinical trials, where FWER needs to be strongly controlled. It is worthy reemphasizing that this unweighted test is valid with small sample size despite the outcome driven nature of the randomization process. The unweighted statistics are appealing in multiple-arm studies where complicated multiplicity adjustment is needed. The unweighted test is also a good alternative to the weighted one if statistics in some stages cannot be computed due to incomplete data. Moderate multiplicity adjusted power gain is available not only in detecting at least one treatment group, but also in identifying and confirming the most efficacious dosing group. Our RABR provides an alternative solution to inferentially seamless Phase II/III design by integrating the identification and confirmation of treatment strategies into one single study. 

The generalization of Theorem \ref{theorem_error} to $m>3$ active treatment arms needs to be cautious, as the condition $r_1 \geq r_2 \geq ... \geq r_m \geq 0$ is not sufficient. Numerical calculation is required to evaluate the multivariate integral, which is beyond the scope of this manuscript. As a compromise, one can apply another set of constraints on the study design to guarantee an analytic type I error protection, for example $r_3 = ... = r_{m-1} = 1/(m+1)$ and $r_1 \geq r_2 \geq r_m \geq 0$. That is to say, the adaptive randomization is only employed on the best two and the worst performing arms. The other $m-3$ arms keep pre-specified randomization probabilities. 

\section*{Acknowledgment}

The authors would like to thank an anonymous associate editor and two anonymous reviewers for their constructive comments, which significantly improved this article. 

This work was supported by AbbVie Inc. AbbVie participated in the interpretation of data, writing, review, and approval of the content of this work. Tianyu Zhan, Ziqian Geng, and Yihua Gu are employees of AbbVie. Lu Cui is a former AbbVie employee and is currently employed by UCB Biosciences. Lanju Zhang is a former employee of AbbVie and is currently employed by Vertex Pharmaceuticals. Ivan S.F. Chan is a former AbbVie employee and is currently employed by Bristol Myers Squibb. All authors may own AbbVie stock.

\section*{Supplementary material}
\label{SM}
Supplementary Materials include the proof of three active treatment groups and of a binary endpoint. The R code to reproduce simulation studies and the case study is available at \url{https://github.com/tian-yu-zhan/RABR_simulations}. An R package $\texttt{RABR}$ is available on Comprehensive R Archive Network (CRAN) to evaluate operating characteristics of the proposed RABR via simulations.

\clearpage

%

\end{document}